\pdfoutput=1
\documentclass[11pt]{article}

\usepackage[T1]{fontenc}
\usepackage[utf8]{inputenc}
\usepackage{lmodern}
\usepackage[letterpaper,margin=1in]{geometry}

\usepackage{amsmath}
\usepackage{amssymb}
\usepackage{amsthm}
\usepackage{graphicx}
\usepackage{booktabs}
\usepackage{longtable}
\usepackage{array}
\usepackage[labelfont=bf,font=small]{caption}
\usepackage{microtype}
\usepackage[hidelinks,breaklinks=true]{hyperref}
\usepackage{url}

\hypersetup{
  pdftitle={Compute-Bounded Security Assurance: Coverage, Verification, and
            Response under Resource Constraints},
  pdfsubject={Theoretical framework for defensive security assurance},
  pdfauthor={jithinvg@bud.studio; dittops@bud.studio; Bud Ecosystem}
}

\theoremstyle{plain}
\newtheorem{proposition}{Proposition}

\newcommand{\E}{\mathbb{E}}
\newcommand{\Prob}{\mathbb{P}}
\newcommand{\ind}{\mathbf{1}}
\DeclareMathOperator{\Var}{Var}
\DeclareMathOperator{\Corr}{Corr}
\newcommand{\iidsim}{\stackrel{\mathrm{iid}}{\sim}}
\newcommand{\eff}{\mathrm{eff}}
\newcommand{\KV}{\mathrm{KV}}
\newcommand{\TV}{\mathrm{TV}}
\newcommand{\KL}{\mathrm{KL}}

\newcolumntype{L}[1]{>{\raggedright\arraybackslash}p{#1}}

\title{%
  \LARGE Compute-Bounded Security Assurance\\[0.45em]
  \large Coverage, Verification, and Response under Resource Constraints}
\author{%
  \texttt{jithinvg@bud.studio}\quad\texttt{dittops@bud.studio}\\
  Bud Ecosystem}
\date{September 2026}

\begin{document}
\maketitle

\begin{abstract}
Additional inference compute can increase the number of correctly resolved
security-assurance tasks, but repeated success, unique coverage, accepted
evidence, and operational protection are different quantities. We develop a
resource-constrained framework that separates them. For repeated conditionally
independent attempts with latent success probability $\Theta$, coverage is
$C_n = 1 - \E[(1-\Theta)^n]$, and its limiting value is
$1 - \Prob(\Theta = 0)$. Positive pairwise outcome correlation does not by
itself imply a ceiling below one: we construct two models with the same mean
success and pairwise correlation but different limiting coverage. We
distinguish this result from the effective sample size used to estimate a
mean, and show why finite-budget observations cannot generally identify an
asymptotic support ceiling. We then connect coverage to fallible evidence
checking, proper scoring of factual grounding, complete resource accounting,
service capacity, and a response model that includes mitigation delay. A
conceptual defensive architecture separates evidence analysis, adjudication,
and operational authority. An evaluation protocol specifies held-out tasks,
paired comparisons, negative cases, and uncertainty reporting. The
contribution is a consistent theoretical synthesis and a set of
counterexamples to invalid extrapolations, rather than an empirical scaling
law. All numerical illustrations are analytic; no model-parity result,
hardware benchmark, or general attacker--defender equilibrium is claimed.
\end{abstract}

\section{Introduction}

Language-model systems can assist with source interpretation, specification
review, configuration analysis, and remediation assessment. Their utility
depends on the evidence made available, the correctness of their conclusions,
and the time and cost required to validate and use those conclusions. More
generated text is therefore an ambiguous measure of security progress. A
system can produce more reports while resolving fewer distinct obligations,
burdening reviewers, or missing a deployment deadline.

Evidence from repeated-sampling studies shows that additional inference can
improve task coverage under particular models, benchmarks, and selection
procedures \cite{brown2024monkeys,snell2024testtime,wu2024inference}.
Security-specific demonstrations and the DARPA AI Cyber Challenge establish
that automated analysis and repair can be useful in realistic software
settings \cite{darpa2025aixcc}. Those observations motivate a
resource-allocation question. They do not establish a universal relationship
between FLOPs and security, or between pairwise correlation and the attainable
fraction of tasks.

This paper develops a framework for \emph{defensive security assurance}: the
production and validation of evidence about specified properties of authorized
systems. Its units are explicit obligations and adjudicated outcomes. We
analyze repeated resolution attempts as abstract random trials; the paper does
not specify an autonomous exploitation procedure. An external harmful-event
process appears only to quantify the timing requirements of defensive
response.

\paragraph{Contributions.}
First, we separate unverified obligations from actual violations and state the
assumptions needed for any residual-mass claim. Second, we derive a coverage
model with an explicit support ceiling and demonstrate that mean success and
pairwise correlation do not determine that ceiling. Third, we distinguish
correct outcomes from accepted reports and give error-aware accounting for
repeated verification. Fourth, we connect complete budget and deadline
constraints to service capacity and response latency. Finally, we provide a
conceptual defensive architecture and a statistical evaluation protocol
consistent with these distinctions.

The probability identities, submodular coverage results, and
constrained-optimization principles used here have established foundations.
Our contribution is their integration into one assurance model, together with
explicit counterexamples that prevent unjustified extrapolation. We make no
priority claim for these underlying mathematical tools and report no new
empirical performance results.

\section{Scope and related foundations}

Formal verification establishes properties relative to a model and its
assumptions. The seL4 project illustrates that substantial guarantees can be
proved for concrete systems; its documented assumptions also illustrate why
deployment boundaries must be stated precisely
\cite{sel4verification,sel4assumptions}. Rice's theorem concerns uniform
decision procedures for nontrivial extensional properties of partial
computable functions \cite{rice1953}. It neither proves that a particular
program is unsafe nor rules out broad classes of machine-checkable proofs.

Attack-graph analysis provides a representation of security-relevant
dependencies under a specified network model \cite{singhal2011attackgraphs}.
Such representations are useful for organizing obligations, but their
existence does not imply that every additional interface increases reachable
behavior or creates a violation. Empirical reports on Android's memory-safety
improvements support prevention of particular defect classes
\cite{vanderstoep2024memorysafety}. Their denominators are reported
vulnerabilities and code-related measures, not a uniform probability measure
over all reachable states.

Repeated-sampling work studies the fraction of benchmark problems solved by at
least one sample \cite{brown2024monkeys}. Test-time allocation and
inference-scaling studies demonstrate workload-specific tradeoffs between
model choice and additional inference
\cite{snell2024testtime,wu2024inference}. These results motivate measuring
complete budget curves. They do not imply that every smaller model can match a
larger one, or that a shared prompt makes independently sampled outputs
statistically dependent conditional on the prompt. Conditional and
population-level dependence must be distinguished.

Expected union coverage is a classical monotone submodular objective
\cite{nemhauser1978}. Adaptive selection requires additional conditions, such
as adaptive submodularity, to inherit comparable guarantees
\cite{golovin2011adaptive}. Determinantal point processes describe repulsive
subset sampling \cite{kulesza2012dpp}; diversity under a chosen representation
is not automatically diversity of correct outcomes.

Inference systems supply useful mechanisms for improving service efficiency.
DistServe separates prefill and decoding to manage interference; SGLang
enables reuse across structured model programs; PagedAttention manages
attention-cache memory; and Mooncake organizes serving around distributed
cache state
\cite{zhong2024distserve,zheng2023sglang,kwon2023pagedattention,qin2024mooncake}.
These are component results. A security-assurance system must additionally
account for evidence processing, adjudication, failed work, and response
delay. We do not transfer reported component speedups into an assumed
end-to-end improvement.

Proper scoring rules provide a basis for evaluating factual probabilistic
predictions \cite{gneiting2007scoring}. Classical constrained optimization
gives allocation conditions under stated regularity and convexity assumptions
\cite{boyd2004convex}. Security-investment models likewise depend on
particular loss and response functions \cite{gordon2002economics}; their
investment bounds cannot be moved unchanged between different decision
problems.

\section{Obligations, residual uncertainty, and violations}

\subsection{A bounded and explicit object of assurance}

Fix a system version, configuration, property specification, environmental
assumptions, and evaluation horizon. Let $\Omega$ be a finite collection of
assurance obligations. An obligation may concern a state invariant, a bounded
trace, or a specified relation between executions. These are different
semantic objects; a property involving multiple executions must be represented
as such, rather than silently reduced to a predicate on one state.

Assign each obligation $j$ a fixed nonnegative weight $w_j$, and write
$\mu(A) = \sum_{j \in A} w_j$ for $A \subseteq \Omega$. For descriptive
coverage, weights may sum to one. For workload accounting they may be
unnormalized. A risk interpretation requires additional assumptions about
event probabilities, consequences, and overlap; arbitrary severity scores do
not supply those assumptions.

Let $V \subseteq \Omega$ be the obligations that are actually violated. Let
$P \subseteq \Omega \setminus V$ be those proved satisfied under the declared
assumptions. Then
\begin{equation}
  U = \Omega \setminus P, \qquad V \subseteq U, \qquad
  \mu(U) = \mu(\Omega) - \mu(P).
\end{equation}
The word \emph{proved} matters. An inconclusive test does not add an
obligation to $P$. Nor does a violation enter $P$ merely because it has been
discovered. Discovery changes knowledge; successful remediation changes the
system and hence the obligation instance being evaluated. Since
$P \cap V = \varnothing$, $U \cap V = V$ for a fixed system under a sound proof
model.

\subsection{What composition does and does not imply}

If $n$ finite-state components have local state counts $s_1, \dots, s_n$,
their composed reachable state set $R$ satisfies
\begin{equation}
  |R| \le \prod_{i=1}^{n} s_i .
\end{equation}
This is an upper bound. A family of components with two syntactic states but
only one reachable state each has $|R| = 1$ even though the product bound is
$2^n$. A synchronization constraint can also restrict a product space to a
small diagonal. Adding a restrictive component can reduce reachable behavior.
No lower bound, monotonicity statement, or violation-density claim follows from
equation~(2) alone.

Similarly, there is no general sublinear law relating proof-covered state count
to verification budget. A short inductive invariant may establish a property
for a very large or infinite family of states. Explicit enumeration and
symbolic or compositional proofs have different cost structures. A counting
identity is not a complexity lower bound.

For a sequence of systems with comparable unnormalized measures, an actual
growth result requires actual growth assumptions. For example, if
$\mu(\Omega_n) \to \infty$ and $\mu(P_n) = o(\mu(\Omega_n))$, then
equation~(1) implies $\mu(U_n) \to \infty$. The conclusion is conditional on
both premises. It does not apply to a normalized measure, which never exceeds
one, and a cross-version comparison requires a declared mapping of obligations.

\subsection{A legitimate conditional density statement}

\begin{proposition}[Conditional residual expectation]
Let $\mathcal{E}$ denote the evidence available after a sound assessment, let
$U$ and its nonnegative weights be $\mathcal{E}$-measurable, and let
$Z_j = \ind\{j \in V\}$. If, almost surely,
$\Prob(Z_j = 1 \mid \mathcal{E}) \ge \delta$ for every $j \in U$, then
\begin{equation}
  \E\!\left[\sum_{j \in U} w_j Z_j \;\middle|\; \mathcal{E}\right]
  \;\ge\; \delta\,\mu(U).
\end{equation}
\end{proposition}

\begin{proof}
Conditional linearity of expectation gives
$\sum_{j \in U} w_j \Prob(Z_j = 1 \mid \mathcal{E})$, and each summand is at
least $w_j \delta$. No independence assumption is needed.
\end{proof}

The probability is over a population of systems or explicitly modeled
uncertainty. The premise is a strong \emph{post-assessment} condition: an
unconditional average defect density does not establish it for the selectively
unverified set. It is not justified by a percentage of reported CVEs belonging
to one defect class. We do not assume such a lower bound in the coverage
analysis that follows. Neither a positive expectation nor a large $U$ proves a
violation exists in every instance.

Undecidability also supplies no missing density premise. A total sound and
complete decision procedure for all programs and an appropriate nontrivial
semantic property is impossible in the classical setting \cite{rice1953}.
Sound incomplete verification, decidable restricted languages, finite-state
checking, and proofs for specific programs remain possible. Assurance should
combine proof and empirical evidence where each is applicable.

\section{Sequential assessment and evidence quality}

\subsection{A decision process with explicit uncertainty}

For discrete state and observation spaces, a finite-horizon partially
observable model may be written
\begin{equation}
  \mathcal{M} = (\mathcal{X}, \mathcal{A}, \mathcal{O}, P, O, r, b_0, H),
\end{equation}
where $P(x' \mid x, a)$ is a transition kernel, $O(o \mid x', a)$ an
observation kernel, $r$ the task reward, and $b_0$ an initial belief. Time,
resource consumption, and relevant assessor state can be included in
$\mathcal{X}$. The policy is separate from the environment tuple. A belief
update, when its normalizer is positive, is
\begin{equation}
  b_{t+1}(x') =
  \frac{O(o_{t+1} \mid x', a_t) \sum_{x} P(x' \mid x, a_t) b_t(x)}
       {\sum_{y} O(o_{t+1} \mid y, a_t) \sum_{x} P(y \mid x, a_t) b_t(x)} .
\end{equation}
This is a modeling language, not a claim that an LLM implements exact Bayesian
filtering. A history-state MDP is another representation when the history
contains the information needed to predict future observations. Either
formulation must retain stochastic tool behavior and incomplete evidence when
they matter.

In the defensive scope considered here, actions include reviewing authorized
evidence, checking specified configuration invariants, and assessing
documented remediation results. Completion means a correct adjudicated
conclusion about an obligation. Intermediate activity measures can be useful
diagnostics, but optimizing them need not preserve the terminal objective.

\subsection{Acceptance is distinct from truth}

Let $Z \in \{0,1\}$ indicate whether a proposed conclusion is correct, and let
$A \in \{0,1\}$ indicate checker acceptance. Within a specified evaluation
population, define
\begin{equation}
  \pi = \Prob(Z = 1), \quad
  v = \Prob(A = 1 \mid Z = 1), \quad
  f = \Prob(A = 1 \mid Z = 0).
\end{equation}
For nonzero acceptance probability,
\begin{equation}
  \Prob(A = 1) = \pi v + (1-\pi) f, \qquad
  \Prob(Z = 1 \mid A = 1) = \frac{\pi v}{\pi v + (1-\pi) f} .
\end{equation}
Thus $\pi v$ counts correctly accepted proposals, not all accepted proposals.
A proof checker may establish soundness relative to a formal statement and
trusted implementation. A crash, differential discrepancy, or model judgment is
not automatically a sound certificate of an arbitrary security property.
Specification adequacy, environment fidelity, and implementation correctness
remain distinct obligations. Validation may also be expensive, especially for
distributed behavior and deployment effects.

With $\pi = 0.01$, $v = 0.90$, and $f = 0.01$, the positive predictive value is
$0.476190$. This analytic example shows why high sensitivity alone does not
make accepted reports reliable. Actual base rates and conditional error rates
must be measured on the evaluated population, including hard negative cases
and abstentions.

\begin{proposition}[Repeated checking error budget]
Suppose a campaign makes at most $N$ checking decisions. Let $F_i$ be the event
that decision $i$ accepts an incorrect conclusion, and let $\mathcal{H}_{i-1}$
be the preceding history. If
$\Prob(F_i \mid \mathcal{H}_{i-1}) \le \alpha_i$ almost surely, for fixed
nonnegative $\alpha_i$, then
\begin{equation}
  \Prob\!\left(\bigcup_{i=1}^{N} F_i\right) \le \sum_{i=1}^{N} \alpha_i .
\end{equation}
\end{proposition}

\begin{proof}
Take expectations of each conditional bound and apply the union bound. No
independence is required. Decisions not made may be assigned empty error
events.
\end{proof}

The premise is stronger than a favorable average false-acceptance rate on a
static validation set. It must hold under the adaptive histories being
assessed. Where it cannot be defended, the bound is a specification target
rather than a measured guarantee. Repeating two checkers with shared failure
modes does not justify multiplying their error probabilities.

\subsection{Grounding and decision quality}

For a labeled fact $Z$ with true distribution $p^*$ and predicted distribution
$q$, the useful KL direction is $D_{\KL}(p^* \Vert q)$. If truth is the point
mass at $z^*$,
\begin{equation}
  D_{\KL}(\delta_{z^*} \Vert q) = -\log q(z^*).
\end{equation}
By contrast, $D_{\KL}(q \Vert \delta_{z^*})$ is infinite whenever $q$ assigns
positive mass to any false label. A difference of such infinities does not
define progress.

For a preregistered finite set of labeled facts, evaluate before-and-after
predictions $q_0, q_1$ using
\begin{equation}
  G = \frac{\E[\log_2 q_1(Z) - \log_2 q_0(Z)]}{\E[T_{\mathrm{ref}}]}, \qquad
  I = \frac{\E[u(d_1, Z) - u(d_0, Z)]}{\E[T_{\mathrm{ref}}]} .
\end{equation}
Here $T_{\mathrm{ref}} > 0$ is processed text under a fixed reference
tokenizer, and $u$ is a declared decision utility. Use positive-support
predictive distributions, or disclose a prespecified smoothing convention. If
separately scored binary facts are aggregated, report the sum of marginal
scores; it is not a joint likelihood without an appropriate joint model. A
bounded proper score, such as the Brier score, is an alternative when
logarithmic losses are unsuitable \cite{gneiting2007scoring}.

The two rates have different units and can be negative. They cannot generally
be multiplied, minimized together, or interpreted as intrinsic intelligence.
Under a true Bayesian joint model and exact prior and posterior predictions,
expected log-score gain is mutual information. For approximate predictions it
is simply a predictive-score difference. Identical evidence availability is not
identical evidence use; factorial comparisons are needed to separate model,
evidence-interface, and interaction effects.

Participation and correctness should likewise be separated. For an attempt
indicator $E$, the exact chain rule is
\begin{equation}
  \Prob(E = 1, Z = 1, A = 1)
  = \Prob(E = 1)\,\Prob(Z = 1 \mid E = 1)\,\Prob(A = 1 \mid E = 1, Z = 1).
\end{equation}
A change in participation does not establish an improvement in conditional
correctness. We therefore avoid a universal scalar that treats refusal, missing
information, competence, and acceptance as interchangeable forms of capability
suppression.

\section{Coverage under repeated assessment}

\subsection{Define the unit before fitting a curve}

Fix a finite evaluation universe $\mathcal{J}$ of distinct outcomes with
weights $w_j \ge 0$ and $W = \sum_j w_j > 0$. Let $Y_{ij} = 1$ mean that
assessment attempt $i$ correctly resolves outcome $j$ and that this resolution
meets the declared evidence standard. Set
\begin{equation}
  C_n = \frac{1}{W} \sum_{j \in \mathcal{J}} w_j\,
        \Prob\!\left(\bigcup_{i=1}^{n} \{Y_{ij} = 1\}\right).
\end{equation}
No independence between different outcomes is required for this identity. It
differs from the chance of finding at least one issue in an entire system and
from the number of accepted reports. It is directly measurable only when the
evaluation universe or an appropriate sampling frame is known. On open-ended
production work, report distinct adjudicated outcomes; do not call their count
a fraction of all unknown vulnerabilities.

If attempts are independent and identically distributed conditional on outcome
$j$, with $\Prob(Y_{ij} = 1 \mid j) = p_j$, then
\begin{equation}
  C_n = \frac{1}{W} \sum_{j} w_j \left[1 - (1-p_j)^n\right].
\end{equation}
The common-rate curve $1 - (1-p)^n = 1 - e^{-n\lambda}$ is exact when
$p_j = p$ for all weighted outcomes and $\lambda = -\log(1-p)$. In general,
$\lambda$ is not equal to the one-attempt coverage fraction; replacing it by
$p$ is a small-$p$ approximation.

Equivalently, draw a weighted random outcome and let its latent per-attempt
success probability be $\Theta \in [0,1]$. Additional run-level shared
conditions can be incorporated in $\Theta$ when conditional independence
remains defensible. Then
\begin{equation}
  C_n = 1 - \E[(1-\Theta)^n].
\end{equation}
This model describes fixed-procedure repetition. It is not a model of arbitrary
adaptive reassessment that changes its information or procedure after each
attempt.

\begin{proposition}[Coverage limit and diminishing increments]
Under equation~(14),
\begin{align}
  C_\infty := \lim_{n \to \infty} C_n &= 1 - \Prob(\Theta = 0), \\
  C_{n+1} - C_n &= \E[\Theta(1-\Theta)^n] \ge 0, \\
  (C_{n+2} - C_{n+1}) - (C_{n+1} - C_n)
    &= -\E[\Theta^2 (1-\Theta)^n] \le 0.
\end{align}
\end{proposition}

\begin{proof}
The bounded integrand $(1-\Theta)^n$ converges pointwise to
$\ind\{\Theta = 0\}$. Bounded convergence gives equation~(15). Subtracting
consecutive instances of equation~(14) gives the other identities.
\end{proof}

A strict support ceiling therefore concerns positive probability of \emph{zero}
success under the fixed procedure. Very small positive success probabilities
can create practical budget or deadline barriers while leaving
$C_\infty = 1$. These are different claims.

For $n \ge 2$, $(1-x)^n$ is convex on $[0,1]$. Jensen's inequality consequently
gives
\begin{equation}
  C_n \le 1 - (1-q)^n, \qquad q = \E[\Theta].
\end{equation}
Replacing heterogeneous difficulty by its mean is optimistic for repeated task
coverage, even when all attempts are independent conditional on their task.

\subsection{Why an effective sample size is not a coverage law}

Let $X_1, \dots, X_n$ be Bernoulli outcomes for a randomly drawn task or shared
condition, with common mean $q \in (0,1)$ and common pairwise correlation $c$.
Their sample mean satisfies
\begin{equation}
  \Var(\overline{X}_n) = \frac{q(1-q)}{n}\left[1 + (n-1)c\right], \qquad
  n_{\eff} = \frac{n}{1 + (n-1)c} .
\end{equation}
This is a variance identity. It compares precision of mean estimation with
independent sampling. It does not determine
$\Prob(X_1 = \cdots = X_n = 0)$, which depends on higher-order joint
probabilities. In particular,
\begin{equation}
  \widetilde{C}_n = 1 - (1-q)^{n_{\eff}}
\end{equation}
is an additional response-curve assumption, not a consequence of
equation~(19). Its limit cannot be advertised as a theorem about all positively
correlated assessments. Negative correlations are possible for finite
collections subject to feasibility constraints; we focus below on
$0 < c < 1$, where the claimed ceiling already fails.

\begin{proposition}[Equal pairwise correlation, different coverage limits]
For any $q, c \in (0,1)$ there are two infinitely extendible, conditionally iid
Bernoulli models with mean $q$ and pairwise correlation $c$, one having
coverage limit one and the other having coverage limit strictly below one.
\end{proposition}

\begin{figure}[t]
  \centering
  \includegraphics[width=\textwidth]{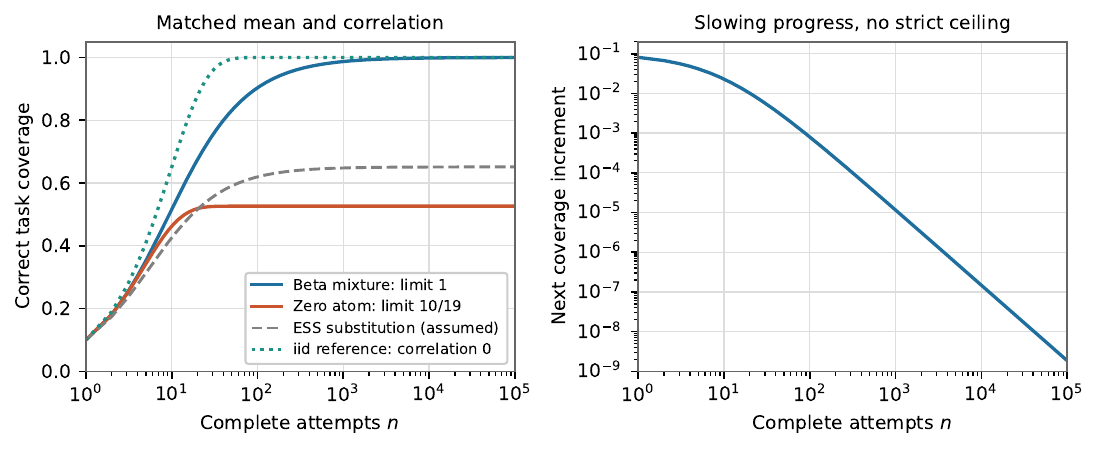}
  \caption{Analytic coverage curves. Left: the beta and zero-atom models share
  $q = c = 0.1$ but have different limits; the effective-sample-size
  substitution matches neither. The iid reference has the same mean and zero
  correlation. Right: the diminishing increments of the beta model remain
  positive despite its limit being one. Values are evaluations of the displayed
  formulas, not model runs.}
  \label{fig:coverage}
\end{figure}

\begin{proof}
In a mixture with $X_i \mid \Theta \iidsim \mathrm{Bernoulli}(\Theta)$,
\begin{equation}
  \E[X_i] = \E[\Theta], \qquad
  \Corr(X_i, X_j) = \frac{\Var(\Theta)}{q(1-q)}, \quad i \ne j.
\end{equation}
For the first model take $\Theta \sim \mathrm{Beta}(a,b)$ with
\begin{equation}
  a = q(c^{-1} - 1), \qquad b = (1-q)(c^{-1} - 1).
\end{equation}
It has mean $q$, variance $cq(1-q)$, and no atom at zero. Its coverage is
\begin{equation}
  C_n^{\mathrm{beta}} = 1 - \frac{\mathrm{B}(a, b+n)}{\mathrm{B}(a,b)}
  \longrightarrow 1 .
\end{equation}
For the second model set $r = q + c(1-q)$ and let $\Theta = r$ with probability
$q/r$, and $\Theta = 0$ otherwise. Its mean is $q$ and its variance is
$q(r-q) = cq(1-q)$. But
\begin{equation}
  C_n^{\mathrm{atom}} = \frac{q}{r}\left[1 - (1-r)^n\right]
  \longrightarrow \frac{q}{r} < 1 .
\end{equation}
Thus both means and all pairwise correlations match, while the limiting
coverage differs.
\end{proof}

For $q = c = 0.1$, the models are $\mathrm{Beta}(0.9, 8.1)$ and a mixture
supported on $\{0, 0.19\}$. Their limits are respectively $1$ and
$10/19 \approx 0.526316$. The substituted curve in equation~(20) predicts
$1 - 0.9^{10} \approx 0.651322$, which is neither limit. Figure~1 plots the
exact curves. It also makes clear that diminishing increments occur both with
and without a strict ceiling.

\subsection{What two moments can establish}

Pairwise information can bound coverage without identifying it. With
$S_n = \sum_i X_i$, Cauchy--Schwarz gives
$(\E S_n)^2 \le \Prob(S_n > 0)\E[S_n^2]$. Substituting the equicorrelated
second moment yields
\begin{equation}
  \frac{nq}{1 + (n-1)\left[q + c(1-q)\right]} \le \Prob(S_n > 0)
  \le \min\{1, nq\} .
\end{equation}
The upper bound is the union bound. These inequalities hold whenever the
stipulated Bernoulli joint distribution exists and the denominator is positive;
they do not require a latent mixture.

Within the conditional-iid mixture model, applying Cauchy--Schwarz to
$\Theta\ind\{\Theta > 0\}$ also gives
\begin{equation}
  \frac{q}{q + c(1-q)} \le C_\infty \le 1 .
\end{equation}
The two constructions in proposition~4 attain these endpoints. This is the
sharp interval for the limiting coverage obtainable from these two moments
within that model class. It is an interval of possible limits, not a predicted
ceiling.

\subsection{Finite observations cannot certify inaccessible support}

\begin{proposition}[A finite-budget indistinguishability bound]
Consider a latent success distribution with mass $\pi_0$ at zero. Replace that
mass by an atom at $\varepsilon \in (0,1]$, leaving the remaining distribution
unchanged. Let $P_n$ and $Q_n$ be the respective laws of the full vector of $n$
conditionally iid Bernoulli outcomes. Then their total variation distance
satisfies
\begin{equation}
  \lVert P_n - Q_n \rVert_{\TV}
  \le \pi_0\left[1 - (1-\varepsilon)^n\right] \le \pi_0 n \varepsilon .
\end{equation}
Nevertheless, moving this atom raises limiting coverage by $\pi_0$.
\end{proposition}

\begin{proof}
Couple the mixture-component draw. Off the moved atom, use identical Bernoulli
draws. On it, one model always returns zero; the other differs only if at least
one $\mathrm{Bernoulli}(\varepsilon)$ draw is one. This bounds the probability
that the coupled vectors differ, hence bounds total variation. The last
inequality is the union bound. The limiting statement follows from
proposition~3.
\end{proof}

For any fixed number of observations, the distributions can therefore be made
arbitrarily close while their asymptotic coverage differs substantially.
Finite-budget data can support a practical coverage limit over the observed
budget range. Identifying exact zero-probability mass requires additional
structural or parametric assumptions, which must be reported rather than hidden
in an asymptote fit.

\subsection{Adaptive attempts and coverage selection}

For one outcome, let $h_i$ be the success probability on attempt $i$
conditional on all earlier attempts failing. The chain rule gives
\begin{equation}
  C_n = 1 - \prod_{i=1}^{n} (1 - h_i).
\end{equation}
This identity does not assume independent attempts. If $h_i \ge \epsilon > 0$
for all $i$, coverage tends to one. More generally, a divergent sum of $h_i$ is
sufficient, since $\log(1-h_i) \le -h_i$. Independence is therefore not
necessary for eventual success. The failure-conditioned probabilities must be
defined on histories of positive probability; after certain success, coverage
is already one.

For a fixed candidate set $\mathcal{A}$ and jointly defined random outcome sets
$B_a$, the function
\begin{equation}
  F(S) = \E\!\left[\sum_j w_j \ind\{j \in \cup_{a \in S} B_a\}\right],
  \qquad S \subseteq \mathcal{A},
\end{equation}
is monotone submodular. For each realization, adding a candidate to a larger
union can only reduce its new contribution; taking expectations preserves this
inequality. The classical $1 - 1/e$ greedy guarantee applies to normalized
monotone submodular maximization under a cardinality constraint and access to
the required marginal values \cite{nemhauser1978}. It does not automatically
transfer to noisy estimated marginals, unequal costs, or candidates whose
output distribution changes with preceding actions.

For example, suppose a uniformly random label $J \in \{1,2\}$ selects the sole
valid conclusion. An information-only assessment reveals $J$ but resolves no
credited outcome by itself; a subsequent assessment specialized to label $1$
resolves it with prior probability $1/2$. After observing $J = 1$, that
assessment's conditional marginal benefit becomes one. Information can increase
later marginal value. Ordinary submodularity of a fixed union does not exclude
this complementarity; adaptive guarantees require their own conditions
\cite{golovin2011adaptive}.

Geometric diversity is also not an outcome-correlation measurement. A valid DPP
requires a positive-semidefinite kernel and an appropriate normalization
\cite{kulesza2012dpp}. It supplies no general guarantee about the coverage
objective in equation~(12). Embedding similarity, duplicate-report frequency,
and Bernoulli outcome correlation should therefore be measured as separate
quantities.

\section{Substitution is conditional on cost, support, and time}

For independent attempts on one fixed task with correctly adjudicated success
probability $p \in (0,1)$, the minimum attempt count for target success
$s \in (0,1)$ is
\begin{equation}
  n^*(s,p) = \left\lceil \frac{\log(1-s)}{\log(1-p)} \right\rceil .
\end{equation}
At $p = 0$, positive success is impossible; at $p = 1$, one attempt suffices.
The familiar approximation $n^* \approx -\log(1-s)/p$ requires small $p$ and
ignores integer rounding.

For $s = 0.95$, $p = 0.1$ requires $29$ attempts, and $p = 0.02$ requires
$149$. The attempt ratio is $149/29 \approx 5.138$. This calculation does not
determine the compute ratio. If a model configuration $m$ has shared cost $b_m$
and per-complete-attempt cost $d_m$, then
\begin{equation}
  B_m(s) = b_m + d_m n^*(s, p_m).
\end{equation}
The cost includes failed attempts and required adjudication. It can be
expressed separately in FLOPs, dollars, or energy with appropriately measured
coefficients. If the lower-success system costs one tenth as much per attempt
and shared costs vanish, its cost ratio in this example is $0.514$, despite
using more attempts. If its attempts cost $1.6$ times as much, the ratio is
$8.221$. A universal super-linear compute premium does not follow from a
probability gap. It requires an explicit family of success and cost functions
and a defined asymptotic variable.

Even in an assumed effective-sample-size response curve, physical compute would
be charged against the raw number of attempts, not $n_{\eff}$. Statistical
discounting does not refund the cost of redundant work. This distinction is
essential whenever reuse or dependence is included in an economic model.

For a homogeneous supported fraction $g$ of tasks, a separate assumed model
gives $C_n = g[1 - (1-p)^n]$. A target $s < g$ then requires
\begin{equation}
  n \ge \left\lceil \frac{\log(1 - s/g)}{\log(1-p)} \right\rceil .
\end{equation}
Targets $s > g$ are impossible; when $0 < p < 1$, $s = g > 0$ is approached but
never reached at finite $n$. Heterogeneous supported tasks require
equation~(14) instead of this two-parameter simplification.

A deadline adds another constraint. In the special case of $m$ identical
available slots, deterministic complete-attempt duration $t > 0$, no startup
overhead, and no other bottleneck, at most $m\lfloor D/t \rfloor$ attempts
finish by time $D$. Shared preprocessing, variable durations, or limited
checking capacity can lower this number. Real substitution claims must
therefore compare the same task distribution and adjudication standard under
explicit budget and deadline constraints. They do not establish parity of
general intelligence or performance on unobserved task classes.

\section{Budgets, service capacity, and delivered outcomes}

\subsection{An outcome objective with explicit constraints}

Let $J_D$ be the distinct correctly adjudicated outcomes completed by deadline
$D$, and define $V_D = \sum_{j \in J_D} w_j$. A useful design problem is
\begin{equation}
  \max_{\pi \in \Pi} \E_\pi[V_D] \quad \text{subject to} \quad
  B_\pi \le B_0, \quad
  \Prob_\pi(\text{any false acceptance}) \le \delta,
\end{equation}
with the budget constraint imposed pathwise, or explicitly replaced by an
expected-budget constraint when that is the intended model. Here $\Pi$ is the
set of authorized assessment policies; feasibility also includes operational
and evidence-access restrictions. In practice the error constraint may be an
audited target rather than a certified bound.

Coverage per FLOP is a useful descriptive efficiency measure at a stated
quality and deadline. Maximizing a ratio alone can favor inadequate total
benefit. For example, $(1 - e^{-ax})/x$ decreases with $x > 0$ for $a > 0$,
since its derivative has numerator $e^{-ax}(ax + 1) - 1 < 0$. A ratio objective
would prefer arbitrarily little work if no quality or service requirement were
imposed. Fixed-budget outcome maximization, cost minimization at a required
outcome level, and net benefit maximization are different problems and should
be named accordingly.

FLOPs, elapsed time, dollars, and joules are also different resource measures.
More efficient hardware can reduce elapsed time or energy for the same
arithmetic; it does not change the count of FLOPs already specified by that
arithmetic. It can indirectly change which procedures fit a deadline. Report a
resource vector rather than hiding unknown quantities in one nominal compute
value.

\subsection{Allocation conditions and their limits}

For an illustrative differentiable response $F(g, \ell, n)$ with shared
preprocessing compute $g$, depth $\ell$, breadth $n$, and decode cost
coefficient $\phi > 0$, consider the relaxation
\begin{equation}
  g + n\ell\phi \le B_0 .
\end{equation}
At a regular interior local optimum with a binding budget and no other active
constraints, first-order stationarity gives
\begin{equation}
  F_g = \frac{F_\ell}{n\phi} = \frac{F_n}{\ell\phi} = \mu .
\end{equation}
These equations are marginal conditions, not an algorithm or a guarantee of
global optimality. Breadth is ordinarily integer-valued, the budget expression
is bilinear, and the response need not be concave. If dependence, accuracy, or
service costs vary with allocation, their derivatives belong in $F$ or the
resource function. Boundary optima require complementary slackness. General
KKT principles come from constrained optimization \cite{boyd2004convex}; no
training scaling law establishes the required response surface here.
Appendix~B gives a restricted concave model where a global solution is
justified.

\subsection{Capacity and useful throughput}

Let an admitted assurance case visit stage $i$ an expected $\nu_i > 0$ times,
consume mean service time $s_i > 0$ per visit, and have $m_i > 0$ dedicated
parallel servers there. Flow conservation in a stationary open pipeline
requires $a\nu_i s_i \le m_i$. A design condition with positive utilization
headroom is
\begin{equation}
  a\,\nu_i s_i < m_i \quad \text{for each } i, \qquad
  a < \min_i \frac{m_i}{\nu_i s_i},
\end{equation}
where $a$ is the admitted case-arrival rate. Strict inequality is required for
positive recurrence in an M/M/1 queue; it is not a universal necessity for
every deterministic system. This counts cases, not successful findings. General
networks can require additional stability conditions; shared resources,
synchronization, nonstationary demand, and scheduling must be modeled
separately. Capacities can increase with investment, so a fixed-pipeline
bottleneck is not a universal asymptotic bound.

If each completed case produces at most one unit-weight outcome, a measurable
goodput identity is
\begin{equation}
  Q = \lambda_{\mathrm{complete}}
      \Prob(\text{distinct, correct, relevant, and timely} \mid
            \text{completed}).
\end{equation}
The joint conditional probability can be decomposed into a chain of conditional
factors. Multiplying separately estimated marginal fractions generally gives a
different answer. For weighted or multiple-output cases, use the expected
distinct timely value per completed case instead.

In a synthetic three-stage M/M/1 tandem with Poisson input, independent
exponential service times, service rates $100$, $30$, and $10$ cases/s, and
unit visit counts, stable case throughput has supremum $10$ cases/s. If
$80\%$ of completed cases are distinct and $50\%$ of those are correctly
adjudicated on time, goodput has supremum $4$ outcomes/s. These yields are
stipulated constants for the illustration; actual yield can depend on load. The
limiting input rate of $10$ cases/s is not stable in this model.

If a case-arrival process is Poisson and each case independently produces a
qualifying result with constant probability $p$, thinning gives a Poisson
result process with rate $ap$. For a finite outcome universe, however, the
fraction of results that are \emph{new} typically decreases over time. Neither
a stationary Poisson novelty process nor novelty fraction $= 1 - c$ follows
from pairwise outcome correlation. A $1/(\mu - a)$ mean residence time is
specific to a stable M/M/1 queue; it is not a general pipeline formula.

For a fixed service-work graph of total work $W_s$ and critical-path duration
$L_s$, $m$ identical processors obey $T \ge \max\{W_s/m, L_s\}$. Likewise,
speeding up a fraction $f_s$ of a fixed serial baseline by a factor $\kappa$
gives the ideal Amdahl expression $1/(1 - f_s + f_s/\kappa)$. At $f_s = 0.2$
and $\kappa = 10$ it is $1.219512$. These bounds identify missing service
contributions; they do not establish which stage dominates a particular
organization.

\subsection{Serving and cache accounting}

For conventional dense transformer key/value storage, one sequence requires
approximately
\begin{equation}
  S_{\KV} = 2LNH_{\KV}db \quad \text{bytes},
\end{equation}
where $L$ is layer count, $N$ cached tokens, $H_{\KV}$ the number of key/value
heads, $d$ head dimension, and $b$ bytes per stored element.
Architecture-specific compression, sharding, and metadata alter this estimate.
For $L = 64$, $N = 4096$, $H_{\KV} = 8$, $d = 128$, and $b = 2$, it is exactly
one GiB. At $50$ GiB/s, moving these bytes alone takes $20$ ms.

Cache reuse can avoid repeated prefill computation, but its net gain includes
transfer, conversion, synchronization, eviction, and queueing. If measured
throughput already includes reuse, applying a second reuse multiplier
double-counts the gain. A hit ratio alone does not determine throughput
improvement. Reuse requires compatible weights, model configuration, positional
encoding, precision and cache semantics, as well as tenant and data-access
isolation. Moving between unrelated hosted models is not a shared-forward-pass
operation.

Prefill is often compute-intensive and low-batch decoding often
bandwidth-intensive, but these are workload regimes, not immutable phase laws.
Batch size, sequence length, model architecture, precision, and communication
can change the bottleneck. A phase-disaggregated design should be compared with
a colocated baseline under the same quality and tail-latency requirements
\cite{zhong2024distserve,qin2024mooncake}. No particular accelerator family is
implied by the probability model.

\section{From verified outcomes to defensive protection}

\subsection{A response race with a mitigation interval}

Let $T_H$ be time from a declared episode origin to harmful impact, $T_D$ time
to usable detection, and $M \ge 0$ time from detection to effective mitigation.
Prevention is the event $T_D + M < T_H$. If $T_H$ is independent of
$(T_D, M)$ and $S_H(t) = \Prob(T_H > t)$, then
\begin{equation}
  \Prob(\text{prevent}) = \E[S_H(T_D + M)].
\end{equation}
Without independence, the correct expression uses the conditional survival
probability $\Prob(T_H > T_D + M \mid T_D, M)$. Detection is not prevention,
and evidence acceptance is not effective deployment.

\begin{proposition}[Independent exponential clocks]
Suppose $T_H \sim \mathrm{Exp}(\alpha)$ and $T_D \sim \mathrm{Exp}(\gamma)$ are
independent, with $\alpha, \gamma > 0$. For deterministic mitigation time
$m \ge 0$,
\begin{equation}
  \Prob(\text{prevent}) = \frac{\gamma}{\alpha + \gamma}\,e^{-\alpha m}.
\end{equation}
If mitigation time is random and independent of both clocks, replace
$e^{-\alpha m}$ by $\E[e^{-\alpha M}]$.
\end{proposition}

\begin{proof}
For fixed $m$, integrate the detection density against survival through
mitigation:
\[
  \int_0^{\infty} \gamma e^{-\gamma t} e^{-\alpha(t+m)}\,dt
  = \frac{\gamma e^{-\alpha m}}{\alpha + \gamma}.
\]
Conditioning on independent $M$ gives the second claim.
\end{proof}

\begin{figure}[t]
  \centering
  \includegraphics[width=\textwidth]{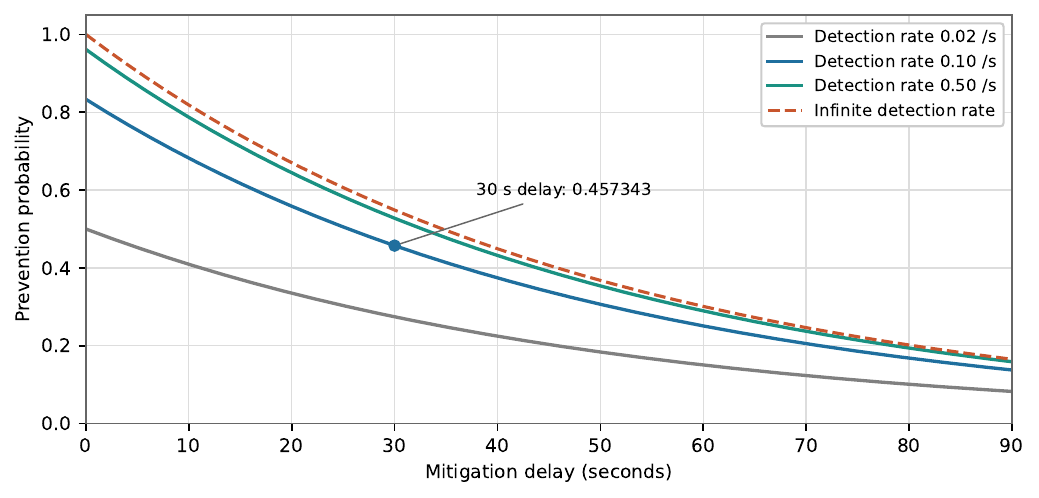}
  \caption{Analytic effect of mitigation delay under independent exponential
  clocks. The curves assume $\alpha = 0.02$ s$^{-1}$ and the displayed
  detection rates. The dashed curve is the infinite-detection-rate limit
  $e^{-\alpha m}$. The marked point gives $0.457343$ at
  $\gamma = 0.10$ s$^{-1}$ and $m = 30$ s.}
  \label{fig:prevention}
\end{figure}

For a finite completion deadline $D \ge m$, the probability of prevention
\emph{and} completed mitigation by $D$ is
\begin{equation}
  \frac{\gamma e^{-\alpha m}}{\alpha + \gamma}
  \left[1 - e^{-(\alpha + \gamma)(D - m)}\right],
\end{equation}
and is zero for $D < m$. These formulas assume effective mitigation succeeds
when completed and omit intervention-caused harm. A failed or partial
intervention requires an extended outcome model.

At $\alpha = 0.02$ s$^{-1}$ and $\gamma = 0.10$ s$^{-1}$, prevention is
$0.833333$ for zero delay and $0.457343$ for $m = 30$ s. As
$\gamma \to \infty$ with $m$ fixed, equation~(40) tends to $e^{-\alpha m}$.
This is a conditional limit on benefits from detection acceleration; it is not
a bound on an adversary's compute advantage. Changing impact hazards,
observability, mitigation effectiveness, or deployment delay changes the
result. Figure~2 illustrates the model without treating its assumed rates as
measurements.

\subsection{Structural controls require a risk model}

Let $j$ index mutually exclusive scenarios with occurrence probabilities
$\pi_j$, losses $\ell_j \ge 0$, and residual harmful-outcome probabilities
$r_j(K) \in [0,1]$ under a control portfolio $K$. Then
\begin{equation}
  L(K) = \sum_j \pi_j \ell_j r_j(K) + L_{\mathrm{intervention}}(K)
\end{equation}
is a stated expected-loss model. Scenario overlap requires joint-event
accounting, a valid partition, or an explicitly labeled bound; summing
overlapping route weights does not generally produce a probability. A
deterministic control with a discharged assurance argument can set an
applicable $r_j$ to zero. A proposed control, incomplete patch, or unsupported
claim cannot be assigned that effect by definition.

Removing weighted risk mass does not by itself multiply per-attempt resolution
probabilities or event hazards by the same fraction. Removing the hardest cases
can increase success among surviving tasks; attention can also shift to the
remaining cases. Coverage normalized against a reduced universe can increase
while absolute harmful opportunity decreases. Consequently there is no general
identity $\lambda(K) = \sigma_K \lambda(\varnothing)$ from a residual-risk
retention fraction $\sigma_K$ alone.

For a horizon $[0, \tau]$, the policy's risk-reduction rate is a causal
estimand,
\begin{equation}
  R(\tau) = \frac{\E[L^0_{[0,\tau]}] - \E[L^1_{[0,\tau]}]}{\tau},
\end{equation}
where both potential outcomes use a common population and horizon, and
policy-induced costs and harm are included. Distinct finding counts do not
identify this difference. Randomized or otherwise justified causal evaluation
is needed; a simulation must be labeled as such.

\subsection{No general equilibrium follows}

For independent exponential discovery clocks, a defender discovers first with
probability $\lambda_D / (\lambda_D + \lambda_A)$. This is a race identity, not
a strategic equilibrium, and it omits post-discovery mitigation. A
game-theoretic equilibrium would additionally require strategy sets,
information, payoff functions, and mutual best responses.

Defenders often have control and observability options unavailable to
outsiders. Attackers may have different information, objectives, timing
choices, and acceptable failure rates. These strategy sets are not naturally
nested. Structural controls can be valuable without implying that defenders
possess a strict superset of all adversary actions, or that a favorable
discovery probability ensures system-wide security.

For the same reason, the Gordon--Loeb investment bound is not a general upper
bound on rational offensive spending \cite{gordon2002economics}. Diminishing
returns alone do not imply a universal $1/e$ fraction. Defensive portfolio
choices should use an explicit expected-loss and cost model, including boundary
solutions and intervention effects.

\section{A conceptual architecture for defensive assurance}

The architecture follows the distinction between information, evidence, and
authority. It is a design specification for authorized assessment and
remediation review, not an implemented or benchmarked system. Figure~3 shows
its trust boundaries. It deliberately does not assume that an embedding score
controls statistical correlation or that a model-generated verdict is sound.

\paragraph{Evidence and assessment.}
Evidence records identify the system version, property, environment, source,
collection time, and permitted use. Observations, hypotheses, and proved claims
remain distinct. Missing observations and conflicting sources remain explicit.
Repository text, retrieved documents, telemetry, and model output are untrusted
data and cannot revise the assessment's authority or scope. The assessor works
within fixed resource and evidence-access limits and can return an unresolved
result.

\begin{figure}[t]
  \centering
  \includegraphics[width=\textwidth]{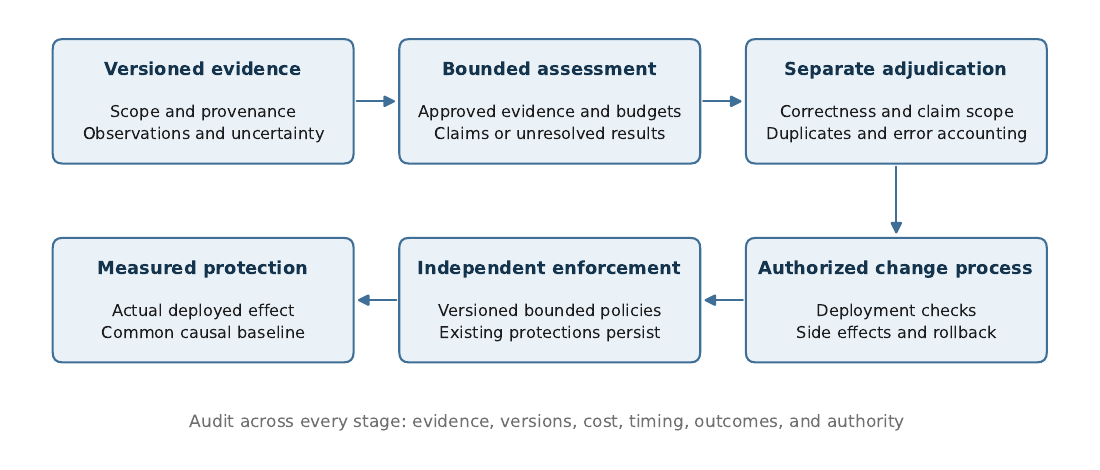}
  \caption{Conceptual defensive architecture. Versioned evidence feeds bounded
  assessment, whose claims undergo separate adjudication. Recommendations reach
  an authorized change process; deployment remains subject to independently
  defined policy. Audit and measurement span the stages. Arrows carry data or
  reviewed recommendations and do not confer authority. No particular model,
  accelerator, or automated search policy is assumed.}
  \label{fig:architecture}
\end{figure}

\paragraph{Adjudication.}
An assessment result states the claim, its scope, supporting evidence, and
limitations. Adjudication distinguishes a formal proof under stated
assumptions, a reproducible bounded observation, a probabilistic judgment, and
an unresolved claim. Duplicate outcomes are reconciled by a declared
equivalence relation with retained evidence, rather than deleted solely because
their wording is similar. Error rates are evaluated on negative as well as
positive cases. A separate reviewer or checker is not presumed statistically
independent merely because it is a separate process.

\paragraph{Authorized remediation and enforcement.}
Adjudicated findings inform an authorized change process. The change record
links the proposed control to affected obligations, expected side effects,
deployment verification, and rollback conditions. Model output does not grant
permission to act. Existing protective controls remain active when inference or
evidence services fail. Time-critical enforcement relies on independently
specified bounded policies; neither fast hardware nor deterministic scheduling
makes a probabilistic model a verified enforcement mechanism. NIST's
adversarial-ML taxonomy supplies relevant integrity and interaction threats for
evaluating this separation \cite{vassilev2025adversarialml}.

\paragraph{Serving and isolation.}
The serving layer records model and runtime versions, complete resource use,
compatibility constraints, and queueing behavior. Cache identity includes
semantic compatibility and access-control scope. Assessment services use
isolated resources where required, and test evidence does not silently become a
production claim. If separate prefill and decoding resources are used, their
benefit must survive the accounting in section~7; a simpler colocated baseline
is required for comparison.

\paragraph{Measurement and change control.}
The audit record links each unique outcome to its evidence, adjudication,
timestamps, resource charges, and any resulting deployment. A model, prompt,
evidence interface, or serving change creates a new evaluated configuration.
Improvements are judged by correctly adjudicated outcomes and deadline
performance, with deployment impact measured separately. Persistent memory is
versioned and invalidated when its premises change. Learned similarity and
internal-model monitoring may inform diagnostics, but neither supplies a proof
of coverage, correctness, or authorization.

\begin{table}[t]
  \centering
  \caption{Architecture claims and the evidence needed to support them.}
  \label{tab:architecture}
  \small
  \begin{tabular}{L{1.55in}L{1.85in}L{2.30in}}
    \toprule
    Proposed benefit & Required comparison &
      Failure that defeats the claim \\
    \midrule
    Better evidence improves assessment &
      Same model, task, and budget; held-out fact and decision scores &
      More context with no gain in correctness, or increased unsupported
      claims \\
    Adjudication improves reliability &
      Positive and negative cases; errors by claim class &
      Shared checker failures or a low positive predictive value \\
    Serving changes improve throughput &
      Same outcome quality and deadline; complete stage costs &
      Transfer, queueing, or review costs erase the gain \\
    Control deployment improves protection &
      Common counterfactual, actual deployment, intervention costs &
      No reduction in expected loss or unacceptable side effects \\
    \bottomrule
  \end{tabular}
\end{table}

\section{Evaluation and reproducibility}

\subsection{Task construction and estimands}

An evaluation should use three complementary settings. First, finite synthetic
obligation sets with exact ground truth permit checking calibration, duplicate
accounting, and abstention, including instances with no violation. Second,
versioned authorized assurance tasks permit specification review, configuration
reasoning, and historical remediation assessment. Third, sanitized operational
replay or shadow-mode recommendations permit measuring timeliness and reviewer
burden. None of these settings is a direct estimate of protection in an
unobserved deployment.

Before running comparisons, declare the task population, system and model
versions, outcome equivalence relation, weights, admissible evidence, complete
budget, deadlines, and primary estimand. Keep tuning and test projects
separate, include time-held-out tasks, and document possible training
contamination. Reference answers used for ground truth must not enter
assessment context. Human reviewers should be blinded to configuration labels
where practical, and disagreements should be retained with their resolution
method.

Define the counted outcome precisely. Correctly establishing that a bounded
obligation is satisfied can be valuable, but it is not a discovered defect.
Report those classes separately before applying any declared common utility.
Unsupported claims, failed runs, timeouts, unavailable evidence, and
abstentions remain visible. For a prespecified intention-to-assess success
endpoint, they remain in its denominator.

\subsection{Coverage and dependence estimation}

Record a task-by-replicate outcome matrix for each frozen configuration. Report
within-task repeated-attempt success separately from variation in difficulty
across tasks. A shared model or prompt does not establish dependence: iid
sampling from a fixed conditional distribution remains iid. Conversely, pooling
heterogeneous tasks can produce positive marginal correlation even when their
attempts are conditionally independent.

Estimate coverage directly over the finite observed budget grid, rather than
inferring it from embeddings. When testing an equicorrelation or latent-mixture
model, compare predicted higher-budget coverage with held-out observations and
report residuals. Pairwise outcome correlations, trajectory similarity,
root-cause overlap, and the fraction of new outcomes are distinct estimands. A
binary correlation is undefined when one variable has zero variance, so
all-success and all-failure strata require separate treatment.

Fit multiple plausible response families where data support them, including
heterogeneous conditional-iid and explicit zero-mass mixtures. Report
sensitivity of extrapolated limits to the chosen family. Proposition~5 means
that a confidence interval for a fitted asymptote is conditional on that
family, not a nonparametric certificate of inaccessible support. An observed
plateau may also reflect censoring, an insufficient budget range, task
heterogeneity, a checker bottleneck, or a change in the evaluated procedure.

\subsection{Quality, resource, and time reporting}

Report at least: correctly adjudicated distinct outcomes by deadline; false
acceptance, missed conclusions, and abstention by task class; grounding score
and decision utility; completed and abandoned cases; median and tail latencies;
and complete resource use. Resource accounts include shared preprocessing,
repeated inference, tool work, adjudication, retries, failed runs, and transfer
overhead. Report human review time separately rather than translating it into
fictitious FLOPs. Undisclosed model compute and reasoning-token counts are
unknown, not zero. Serving price and energy measurements need a measurement
date and method.

Under a fixed task distribution, verified weighted goodput is
\begin{equation}
  Q(\tau) = \frac{1}{\tau}\,
  \E\!\left[\sum_{j \in \mathcal{J}_{\mathrm{distinct}}}
    w_j Y_j \ind\{t_j \le D_j,\; t_j \le \tau\}\right],
\end{equation}
where $Y_j$ indicates an independently adjudicated correct conclusion, $t_j$
its completion time, and $D_j$ its deadline. This quantity is descriptive
productivity. It becomes operational risk reduction only with additional causal
evidence supporting equation~(43).

\subsection{Paired comparisons and noninferiority}

Let $P_A$ and $P_B$ be task-level success probabilities of reference and
candidate configurations under a specified budget and deadline, and let
$\Delta = P_B - P_A$. A noninferiority analysis prespecifies margin
$\epsilon > 0$ and tests $H_0 : \Delta \le -\epsilon$. It supports
noninferiority at level $\alpha$ only when a valid one-sided lower
$(1-\alpha)$ confidence bound exceeds $-\epsilon$. This is a statistical
decision with controlled error under the method's assumptions, not proof that
the true difference satisfies the bound.

Equivalence requires rejecting both $\Delta \le -\epsilon$ and
$\Delta \ge \epsilon$. With compatible one-sided tests at level $\alpha$, a
two-sided $(1 - 2\alpha)$ interval strictly inside $(-\epsilon, \epsilon)$ is
the corresponding criterion. Failure to detect a difference is not equivalence.
Choose margins for practical significance and plan sample size or power before
evaluating systems.

Use paired outcomes on shared tasks and preserve project-level dependence in
uncertainty estimates. A project-cluster bootstrap requires enough independent
projects and is not a cure for very few clusters; use a justified alternative or
report the limitation. Preregister a primary budget and endpoint, or adjust for
multiple comparisons and repeated looks. Equal-dollar, equal-energy, and
equal-deadline comparisons answer different questions; report each frontier
separately. No named-model parity claim is supported by this paper.

\begin{table}[t]
  \centering
  \caption{Empirical questions and observations that would count against a
  proposed benefit.}
  \label{tab:empirical}
  \small
  \begin{tabular}{L{0.55in}L{2.45in}L{2.55in}}
    \toprule
    ID & Question & Evidence against the proposed benefit \\
    \midrule
    E1 & Does additional budget improve correct coverage in the feasible
         range? &
         No held-out improvement after complete cost and deadlines are
         included \\
    E2 & Does better evidence improve factual and decision quality? &
         Unchanged or degraded proper scores and task utility \\
    E3 & Can a cheaper configuration meet the reference service requirement? &
         Failure of the preregistered noninferiority test at feasible
         budgets \\
    E4 & Does serving specialization improve delivered outcomes? &
         Component speedup is erased by queueing, transfer, or
         adjudication \\
    E5 & Does a control reduce operational loss? &
         No identifiable loss reduction after deployment and intervention
         costs \\
    \bottomrule
  \end{tabular}
\end{table}

\subsection{Falsifiable empirical questions}

Table~2 identifies claims that would justify expanding this theoretical
framework. These are proposed studies, not results.

\section{Limitations and implications}

The coverage results are exact within their stated models. Conditional
independence, the latent success distribution, and a fixed evaluation universe
are modeling assumptions, not observations about all security work. Adaptive
procedures can violate the diminishing-increment form, while open-ended
production settings lack a known coverage denominator. The finite-budget
indistinguishability result limits support inference even when the mixture
model is appropriate.

The evidence model separates truth and acceptance but does not solve
specification errors or guarantee independent adjudication. Proper scores need
reliable labels and a defined prediction space. Queueing bounds require a
compatible workload model; the exponential response example assumes independent
clocks, constant hazards, and effective mitigation. Causal risk reduction is
harder to identify than productivity, especially under changing threats and
coupled interventions.

The architecture is a conceptual proposal. No implementation, calibrated
population parameters, hardware comparison, model-performance study, or
deployment trial is reported. The arithmetic and proof checks in the companion
material validate the displayed formulas and examples; they do not validate the
empirical adequacy of the assumptions. The statistical identities should not be
presented as novel probability theory, and the framework's practical value
remains an empirical question.

The defensive implication is nevertheless useful: invest in evidence quality,
soundly scoped verification, effective controls, and timely response according
to their measured contribution. Additional compute is one input to that
decision. A reduction in actual harmful opportunity can be more valuable than a
higher assessment rate, but no universal ranking follows without costs,
effectiveness, and the evaluated loss model.

\section{Conclusion}

Compute-bounded assurance should be evaluated through correctly adjudicated,
distinct, timely outcomes under complete resource constraints. Pairwise outcome
correlation determines a variance adjustment, not a universal coverage ceiling.
Under conditional-iid repetition, the true asymptotic ceiling is determined by
zero-probability support, while finite budgets cannot generally distinguish
that support from extremely rare success. Verification error, service capacity,
and mitigation delay then determine how much assessment activity becomes useful
evidence and protection. These distinctions provide a defensible basis for
empirical comparisons without claiming inevitable vulnerability, universal
model substitution, or a general security equilibrium.

\section*{Research transparency}

The numerical illustrations are deterministic evaluations of stated formulas
using synthetic parameter choices. No software vulnerability was reproduced, no
live system was assessed, and no operational performance data were collected
for this paper. The companion source package includes the calculation and
figure-generation scripts. An AI assistant was used in mathematical checking,
literature checking, manuscript revision, and preparation of the
reproducibility material. This assistance is not independent peer review.

\appendix

\section{Notation}

\begingroup
\small
\setlength{\LTcapwidth}{\textwidth}
\begin{longtable}{L{1.65in}L{4.40in}}
  \caption{Symbols and their local scopes.}
  \label{tab:notation}\\
  \toprule
  Symbol & Meaning and units \\
  \midrule
  \endfirsthead
  \toprule
  Symbol & Meaning and units \\
  \midrule
  \endhead
  \bottomrule
  \endfoot
  $\Omega, U, P, V$ &
    Finite obligation universe, unproved obligations, proved-satisfied
    obligations, and violated obligations \\
  $w_j, \mu$ &
    Nonnegative obligation weights and their additive measure; a risk
    interpretation requires extra assumptions \\
  $\mathcal{E}, Z_j$ &
    Assessment evidence and indicator that obligation $j$ is violated \\
  $\mathcal{X}, \mathcal{A}, \mathcal{O}$ &
    Latent state, admissible action, and observation spaces \\
  $P(x' \mid x, a), O(o \mid x', a)$ &
    Transition and observation kernels; distinct from the proof-covered set
    $P$ \\
  $b_t, H$ &
    Belief at step $t$ and finite decision horizon \\
  $Z, A, E$ &
    Proposal correctness, checker acceptance, and participation indicators \\
  $\pi, v, f$ &
    Proposal base rate, checker sensitivity, and false-acceptance probability;
    policy $\pi$ is distinguished by context \\
  $G, I, T_{\mathrm{ref}}$ &
    Grounding log-score gain per reference token, decision-utility gain per
    reference token, and processed reference-token count \\
  $\mathcal{J}, Y_{ij}, W$ &
    Distinct evaluation outcomes, correct resolution indicator, and total
    outcome weight \\
  $n, C_n$ &
    Completed assessment attempts and weighted fraction correctly resolved at
    least once \\
  $\Theta, q, c$ &
    Latent conditional success probability, its mean, and pairwise Bernoulli
    outcome correlation \\
  $a, b$ in equation~(23) &
    Beta-distribution shape parameters; the arrival rate $a$ below has a
    separate local scope \\
  $\pi_0, r, n_{\eff}$ &
    Zero-success mass, nonzero atom in the counterexample, and
    variance-equivalent sample size \\
  $p, s, n^*, g$ in section~6 &
    Per-attempt success, target success, required integer attempts, and
    supported task fraction \\
  $B_m, b_m, d_m$ &
    Total, shared, and per-attempt cost in one consistently chosen resource
    unit \\
  $B_0, D, V_D$ &
    Budget limit, completion deadline, and distinct correctly adjudicated value
    completed by the deadline \\
  $g, \ell, n, \phi, \mu$ in equation~(35) &
    Preprocessing FLOPs, depth, breadth, FLOPs per depth unit, and budget
    shadow price \\
  $a, m_i, \nu_i, s_i$ in equation~(36) &
    Cases/s, servers, expected visits per case, and seconds per visit \\
  $Q, \lambda_{\mathrm{complete}}$ &
    Correct distinct timely outcomes/s, and completed cases/s \\
  $S_{\KV}, L, N, H_{\KV}, d, b$ &
    KV bytes, layers, cached tokens, KV heads, head dimension, and bytes per
    element \\
  $T_H, T_D, M$ &
    Time to harmful impact, time to usable detection, and mitigation duration,
    all in seconds \\
  $\alpha, \gamma, m$ in equation~(40) &
    Impact and detection rates in s$^{-1}$, and deterministic mitigation
    seconds \\
  $L(K), R(\tau)$ &
    Expected loss under controls $K$, and causal loss reduction per unit
    time \\
  $\Delta, \epsilon, \alpha$ in evaluation &
    Candidate-minus-reference success, noninferiority margin, and test size \\
\end{longtable}
\endgroup

\section{A restricted concave allocation model}

Consider nonoverlapping assurance-benefit classes with allocated budgets
$x_i \ge 0$. Assume their expected benefits are
\begin{equation}
  F_i(x_i) = A_i(1 - e^{-\kappa_i x_i}), \qquad
  A_i \ge 0,\; \kappa_i > 0, \qquad
  \sum_i x_i \le B .
\end{equation}
Here $A_i$ is benefit, $x_i$ and $B$ are in the same resource unit, and
$\kappa_i$ is its inverse. The benefit sum is concave; it is strictly concave
on coordinates with $A_i > 0$. If $B > 0$ and some $A_i > 0$, the optimum uses
all budget on positive-benefit classes. KKT conditions are sufficient, and
funded coordinates satisfy $A_i \kappa_i e^{-\kappa_i x_i} = \eta$ for a common
$\eta > 0$. Thus
\begin{equation}
  x_i^* =
  \begin{cases}
    \kappa_i^{-1}\left[\log(A_i \kappa_i / \eta)\right]_+, & A_i > 0, \\
    0, & A_i = 0,
  \end{cases}
  \qquad \sum_i x_i^* = B .
\end{equation}
The sum is continuous and strictly decreasing in $\eta$ whenever positive, from
infinity at zero to zero at $\max_i A_i \kappa_i$. This determines a unique
multiplier for $B > 0$. If $B = 0$, all allocations vanish. If all $A_i = 0$,
every feasible allocation is optimal and zero spending is a natural choice; the
logarithmic formula need not be evaluated.

The result uses separability, concavity, and the specified exponential
benefits. It is not the same problem as a bilinear depth--breadth budget.
Shared outcomes, deadlines, checker capacity, uncertain intervention effects,
or overlapping causal benefit require a different optimization model. The
general mathematical conditions are standard \cite{boyd2004convex}.

\section{Analytic checks and numerical examples}

Table~4 gives values used in the manuscript. The companion script checks
endpoint cases, exact rational moments of the discrete mixture, integer
minimality of the attempt counts, finite coverage bounds, and numerical
agreement between the response integral and its closed form. It also verifies
the allocation conditions on a synthetic portfolio and includes
counterexamples to the residual-growth and universal-investment claims. These
checks supplement the proofs; they are not formal machine-checked verification
of the whole manuscript.

\begin{table}[ht]
  \centering
  \caption{Reproducible analytic examples. All inputs are synthetic.}
  \label{tab:examples}
  \small
  \begin{tabular}{L{2.95in}L{3.05in}}
    \toprule
    Inputs & Result \\
    \midrule
    $s = 0.95$, $p = 0.1$ or $0.02$ &
      29 or 149 complete attempts \\
    $q = c = 0.1$ &
      Beta-mixture limit 1; zero-atom limit $10/19 = 0.526316$ \\
    $q = c = 0.1$ in ESS substitution &
      Assumed limit $1 - 0.9^{10} = 0.651322$ \\
    $\pi = 0.01$, $v = 0.90$, $f = 0.01$ &
      Accepted-report precision $0.476190$ \\
    M/M/1 tandem rates 100, 30, 10 cases/s &
      Stable admitted rate below 10 cases/s; goodput below 4/s at the stated
      yield \\
    $\alpha = 0.02$, $\gamma = 0.10$, $m = 30$ &
      Prevention probability $0.457343$ \\
    64 layers, 4096 tokens, 8 KV heads, $d = 128$, 2 bytes/element &
      $2^{30}$ bytes; 20 ms at 50 GiB/s, excluding overhead \\
    $f_s = 0.2$, $\kappa = 10$ &
      Ideal Amdahl speedup $1.219512$ \\
    \bottomrule
  \end{tabular}
\end{table}

\paragraph{Diminishing returns do not imply a universal investment fraction.}
For a purely mathematical counterexample, let benefit $f(x)$ for $x \ge 0$
equal $2x$ up to $x = 0.45$, and $1 - 0.1e^{-20(x - 0.45)}$ thereafter. It is
continuously differentiable, nondecreasing and concave, starts at zero, and is
bounded above by one. Net benefit $f(x) - x$ is maximized at
$x^* = 0.45 + \log(2)/20 \approx 0.484657 > 1/e$. Thus bounded concave benefit
alone cannot prove a $1/e$ expenditure bound. A specific economic theorem
requires its own assumptions.

\paragraph{A dependence check at two attempts.}
For any common-mean Bernoulli pair with correlation $c$, direct
inclusion--exclusion gives $C_2 = 2q - q^2 - cq(1-q)$. The beta and zero-atom
constructions have exactly this value. Matching $C_1$ and $C_2$ therefore does
not distinguish their very different large-budget behavior. The companion
material includes the complete figure values so that numerical stability and
plotted scales can be inspected independently.


\begin{thebibliography}{99}
\small

\bibitem{boyd2004convex}
Stephen Boyd and Lieven Vandenberghe.
\newblock \emph{Convex Optimization}.
\newblock Cambridge University Press, 2004.
\newblock \url{https://web.stanford.edu/~boyd/cvxbook/}.

\bibitem{brown2024monkeys}
Bradley Brown, Jordan Juravsky, Ryan Ehrlich, Ronald Clark, Quoc V. Le,
  Christopher R\'{e}, and Azalia Mirhoseini.
\newblock Large language monkeys: Scaling inference compute with repeated
  sampling.
\newblock arXiv:2407.21787, 2024.
\newblock \url{https://arxiv.org/abs/2407.21787}.

\bibitem{darpa2025aixcc}
DARPA.
\newblock AI cyber challenge marks pivotal inflection point for cyber defense.
\newblock DARPA competition results, 2025.
\newblock \url{https://www.darpa.mil/news/2025/aixcc-results}.

\bibitem{gneiting2007scoring}
Tilmann Gneiting and Adrian E. Raftery.
\newblock Strictly proper scoring rules, prediction, and estimation.
\newblock \emph{Journal of the American Statistical Association},
  102(477):359--378, 2007.
\newblock \url{https://doi.org/10.1198/016214506000001437}.

\bibitem{golovin2011adaptive}
Daniel Golovin and Andreas Krause.
\newblock Adaptive submodularity: Theory and applications in active learning
  and stochastic optimization.
\newblock \emph{Journal of Artificial Intelligence Research}, 42:427--486,
  2011.
\newblock \url{https://arxiv.org/abs/1003.3967}.

\bibitem{gordon2002economics}
Lawrence A. Gordon and Martin P. Loeb.
\newblock The economics of information security investment.
\newblock \emph{ACM Transactions on Information and System Security},
  5(4):438--457, 2002.
\newblock \url{https://doi.org/10.1145/581271.581274}.

\bibitem{kulesza2012dpp}
Alex Kulesza and Ben Taskar.
\newblock Determinantal point processes for machine learning.
\newblock \emph{Foundations and Trends in Machine Learning}, 5(2--3):123--286,
  2012.
\newblock \url{https://arxiv.org/abs/1207.6083}.

\bibitem{kwon2023pagedattention}
Woosuk Kwon et al.
\newblock Efficient memory management for large language model serving with
  PagedAttention.
\newblock In \emph{Proceedings of the 29th Symposium on Operating Systems
  Principles}, 2023.
\newblock \url{https://arxiv.org/abs/2309.06180}.

\bibitem{nemhauser1978}
George L. Nemhauser, Laurence A. Wolsey, and Marshall L. Fisher.
\newblock An analysis of approximations for maximizing submodular set
  functions---I.
\newblock \emph{Mathematical Programming}, 14:265--294, 1978.
\newblock \url{https://doi.org/10.1007/BF01588971}.

\bibitem{qin2024mooncake}
Ruoyu Qin et al.
\newblock Mooncake: A KVCache-centric disaggregated architecture for LLM
  serving.
\newblock arXiv:2407.00079, 2024.
\newblock \url{https://arxiv.org/abs/2407.00079}.

\bibitem{rice1953}
H. G. Rice.
\newblock Classes of recursively enumerable sets and their decision problems.
\newblock \emph{Transactions of the American Mathematical Society},
  74(2):358--366, 1953.
\newblock \url{https://doi.org/10.1090/S0002-9947-1953-0053041-6}.

\bibitem{sel4verification}
seL4 Foundation.
\newblock Verification.
\newblock \url{https://sel4.systems/Verification/}, 2026.
\newblock Living documentation, accessed September 7, 2026.

\bibitem{sel4assumptions}
seL4 Foundation.
\newblock What the proofs assume.
\newblock \url{https://sel4.systems/Verification/assumptions.html}, 2026.
\newblock Living documentation, accessed September 7, 2026.

\bibitem{singhal2011attackgraphs}
Anoop Singhal and Xinming Ou.
\newblock Security risk analysis of enterprise networks using probabilistic
  attack graphs.
\newblock Technical Report NIST IR 7788, National Institute of Standards and
  Technology, 2011.
\newblock \url{https://csrc.nist.gov/pubs/ir/7788/final}.

\bibitem{snell2024testtime}
Charlie Snell, Jaehoon Lee, Kelvin Xu, and Aviral Kumar.
\newblock Scaling LLM test-time compute optimally can be more effective than
  scaling model parameters.
\newblock arXiv:2408.03314, 2024.
\newblock \url{https://arxiv.org/abs/2408.03314}.

\bibitem{vanderstoep2024memorysafety}
Jeff Vander Stoep and Alex Rebert.
\newblock Eliminating memory safety vulnerabilities at the source.
\newblock Google Security Blog, September 2024.
\newblock
  \url{https://security.googleblog.com/2024/09/eliminating-memory-safety-vulnerabilities-Android.html}.

\bibitem{vassilev2025adversarialml}
Apostol Vassilev et al.
\newblock Adversarial machine learning: A taxonomy and terminology of attacks
  and mitigations.
\newblock Technical Report NIST AI 100-2e2025, National Institute of Standards
  and Technology, 2025.
\newblock \url{https://doi.org/10.6028/NIST.AI.100-2e2025}.

\bibitem{wu2024inference}
Yangzhen Wu, Zhiqing Sun, Shanda Li, Sean Welleck, and Yiming Yang.
\newblock Inference scaling laws: An empirical analysis of compute-optimal
  inference for problem-solving with language models.
\newblock arXiv:2408.00724, 2024.
\newblock \url{https://arxiv.org/abs/2408.00724}.

\bibitem{zheng2023sglang}
Lianmin Zheng et al.
\newblock SGLang: Efficient execution of structured language model programs.
\newblock arXiv:2312.07104, 2023.
\newblock \url{https://arxiv.org/abs/2312.07104}.

\bibitem{zhong2024distserve}
Yinmin Zhong et al.
\newblock DistServe: Disaggregating prefill and decoding for goodput-optimized
  large language model serving.
\newblock arXiv:2401.09670, 2024.
\newblock \url{https://arxiv.org/abs/2401.09670}.

\end{thebibliography}
\end{document}